\documentclass[letterpaper, 10 pt, conference]{ieeeconf}  

\IEEEoverridecommandlockouts                              

\usepackage{graphics} 
\usepackage{epsfig} 
\usepackage{mathptmx} 
\usepackage{amsmath} 
\usepackage{amssymb}  
\usepackage{xcolor}

\newtheorem{theorem}{Theorem}
\newtheorem{corollary}{Corollary}

\newtheorem{definition}{Definition}

\title{\LARGE \bf
On the Use of the Smith-McMillan Form \\in Decoupling System Dynamics
}

\author{Clarisse Pétua Bosman Barros, Hans Butler, Roland Tóth
\thanks{C.P. Bosman Barros, H. Butler and R. Tóth are with the Eindhoven University of Technology, Department of Electrical Engineering, Control Systems Group, Eindhoven, The Netherlands. Email: {\tt\small c.p.bosman.barros@tue.nl}}
\thanks{R. Tóth is also with the Systems and Control Laboratory, Institute for Computer Science and Control, Kende u. 13-17, H-1111 Budapest, Hungary. H. Butler is also with ASML, 5504 DR Veldhoven, The Netherlands.}
\thanks{This work is part of the TU/e-ASML mini Impulse program ``Advanced piezo-electric wafer stage for lithography and metrology".}
}

\begin{document}

\maketitle
\thispagestyle{empty}
\pagestyle{empty}

\begin{abstract}
In this paper, the use of the Smith-McMillan form in decoupling multiple-input multiple-output system dynamics is analyzed. In short, from a transfer matrix plant model one can obtain a decoupling compensator which leads to a decoupled plant that contains the same transmission poles and zeros of the transfer matrix of the original system. As a result, full decoupling of the plant transfer matrix is obtained for all frequencies, which can be individually addressed by single-input single-output control. The aim of this paper is to present conditions for the decoupled system that guarantee internal stability and performance requirements for the overall control system. Performance specifications are defined in terms of magnitude limits for the maximum singular value of the closed-loop transfer matrices. The potential of the decoupling procedure is shown in a simulation  study of a mechanical system.
\end{abstract}

\vspace{-4mm}
\section{INTRODUCTION}

A conceptually direct approach of controller synthesis for multiple-input multiple-output (MIMO) plants is given by a two-step procedure in which a decoupling compensator is designed first to deal with the interactions in the plant, followed by a diagonal controller's design to the obtained decoupled plant composed of single-input single-output (SISO) systems \cite{skogestad2007multivariable}. The variety of design methods for SISO systems can be used straightforwardly and the overall controller synthesis procedure is simple, but the success of this methodology is subject to the quality the system's decoupling.

There are several methods to obtain a decoupling compensator. The most straightforward procedure is based on the inverse of the transfer matrix of the plant, theoretically decoupling by transforming it into an identity matrix  \cite{skogestad2007multivariable}. However, such a procedure corresponds to the assumption of a perfect plant model and the possibility of a perfect cancellation of all system dynamics. If the plant transfer matrix contains right half complex plane zeros or if it is composed of proper transfer functions, the inverse is unstable or non-proper, respectively, causing implementation issues. These problems are often handled by using various approximations, degrading the decoupling at specific frequency ranges. 

Decoupling can also be based on the singular value decomposition (SVD) of the plant at a fixed given frequency, ensuring that the plant is decoupled in a frequency range in the vicinity of the fixed frequency, given that the decoupling compensator is a static compensator, i.e., composed of constant matrices. For specific cases, such as plants consisting of symmetrical interconnected systems, this procedure leads to a controller structure which is optimal at the fixed frequency \cite{hovd1997svd}. Conversely,  poor decoupling may be obtained for plants without the specific mentioned structure and with a transfer matrix that varies rapidly as a function of frequency.

For specific applications, the decoupling compensator can also be based on the understanding of the physical characteristics of the system, as in \cite{butler2011position} for a six-degrees-of-freedom stage. Decoupling can be achieved even for nonlinear systems, but it requires specific methodologies and technical expertise of the system operation. The design procedure for this approach is usually complicated and time-consuming.

As an alternative one may use the concept of decoupling the system based on the Smith-McMillan form as in \cite{mohsenizadeh2015multivariable}. The decoupled plant contains the same transmission poles and zeros of the original system transfer matrix, thus maintaining its fundamental dynamic characteristics. The transfer matrices are calculated analytically and full decoupling of the plant transfer matrix is obtained for all frequencies. In addition, in \cite{diaz2019pid,mohsenizadeh2015multivariable} the conditions for the properness of the final controller (diagonal controller and decoupling compensator) are obtained. Furthermore, it is shown that if the output signals track the reference signals for the decoupled system, then reference tracking is also guaranteed for the plant with the decoupling compensator and the designed controller. 

In this paper we further investigate how useful this decoupling technique is, specifically regarding the overall internal stability and achievable closed-loop  performance guarantees. The main contributions of this paper are:
\begin{itemize}
\item[(C1)] Transformations of open- and closed-loop transfer matrices between the original and decoupled domains.
\item[(C2)] Sufficient conditions for internal stability with the decoupling based feedback control system. 
\item[(C3)] Showing how performance requirements for the closed-loop translate to specifications for the decoupled system. 
\end{itemize}The paper is organized as follows. The problem statement and the decoupling compensator design procedure are presented in Sections 2 and 3, respectively. In Section 4,  transformations of open- and closed-loop transfer matrices between the original and decoupled system are derived (C1), which are applied to address stability in Section 5 (C2) and to investigate performance requirements in Section 6 (C3). Finally, this decoupling combined with SISO controller design by loop-shaping is applied to an example of a mechanical system, demonstrating advantages and drawbacks of this decoupling methodology for controller synthesis.  

\textit{Notation:} $\mathbb{C}$ denotes the set of complex numbers and $\mathbb{R}$ the set of real numbers. The imaginary unit is denoted by $j = \sqrt{-1}$. The set of real rational proper and stable transfer functions or transfer matrices is denoted by $\mathcal{R}\mathcal{H}_\infty$. The time variable is denoted by $t\in \mathbb{R}$.

\section{PROBLEM STATEMENT} 

Consider a MIMO plant with $n$ inputs and $n$ outputs modeled as a Linear Time-Invariant (LTI) system with $u(t) \in \mathbb{R}^{n\times 1}$ and $y(t) \in \mathbb{R}^{n\times 1}$ representing, respectively, the input and output signals. 
The input-output relation is expressed as 
\vspace{-0.3cm}
 \begin{equation*}
y(s) = P(s)\:u(s), \vspace{-0.1cm} 
\end{equation*}
where $P:\mathbb{C} \rightarrow \mathbb{P} \subseteq \mathbb{C}^{n\times n}$ is assumed to be a known real rational transfer matrix, $s\in \mathbb{C}$ denotes the complex frequency (Laplace variable) and $y(s)$ and $u(s)$ are the Laplace transforms of $y(t)$ and $u(t)$, respectively, on their appropriate region of convergence. The analysis presented here considers a square transfer matrix for the plant, but it can be extended to plants at which the number of inputs and outputs differ. 

A control law defined in the format of a transfer matrix $C(s)$ needs to be designed in order to have the system tracking a set of references $r(t)$. In addition, the control system needs to be such that the closed-loop system is internally stable, the controller is proper and specifications defined for the frequency domain are fulfilled. The connection between the controller and the plant is represented in Figure \ref{fig:blockDiag_controlSys}.

\begin{figure}
      \centering
		\includegraphics[width = 6cm]{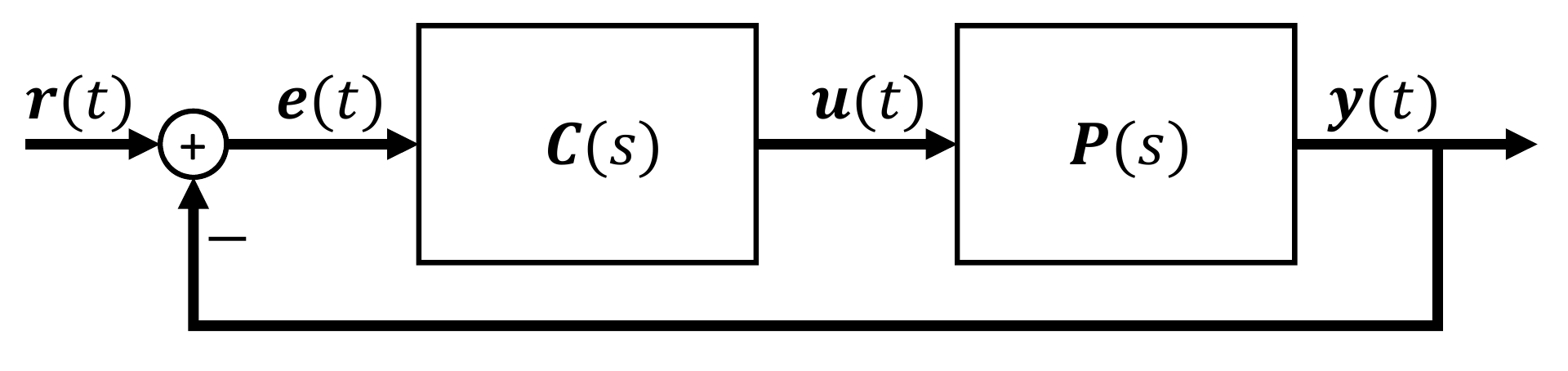} \vspace{-0.2cm}
      \caption{Conventional negative feedback interconnection (original system).}
      \label{fig:blockDiag_controlSys} \vspace{-0.4cm}
\end{figure}


\section{SMITH-MCMILLAN FORM BASED DECOUPLING METHOD}

The Smith-McMillan form is usually used to determine the transmission poles and zeros of a MIMO system \cite{werner2012control}. The decomposition performed to obtain the Smith-McMillan form is done by the multiplication of the transfer matrix by unimodular\footnote{Unimodular polynomial matrices are polynomial matrices whose determinant is a nonzero constant, which implies that they are full rank and their inverse is also a polynomial matrix.} polynomial matrices $U(s)$ and $V(s)$, $\mathbb{C} \rightarrow \mathbb{C}^{n\times n}$:
 \begin{equation}\label{eq:PtransfPSM}
P^{SM}(s) = U(s)\:P(s)\:V(s).\vspace{-0.1cm}
\end{equation} The procedure to find $U$, $V$ and consequently $P^{SM}$ is well defined, widely described in literature (see \cite{macfarlane1976poles}, for instance) and will not be addressed in this paper. 

The Smith-McMillan form $P^{SM}$ of the transfer matrix $P$ is a diagonal transfer matrix. It contains the the transmission poles and zeros of $P$ itself, which are not necessarily the poles and zeros of the individual elements of the matrix. This general definition of poles and zeros better characterizes the properties of MIMO systems (refer to \cite{macfarlane1976poles} for more details).

The decoupling procedure is performed as follows. Consider the matrices $U(s)$ and $V(s)$ to be known. Based on 
\vspace{-0.3cm}
\begin{equation*}
y(s) = P(s)\:u(s) \Rightarrow U(s)\:y(s) = U(s)\:P(s)\:u(s), \vspace{-0.1cm} 
\end{equation*} 
if we define $ u^{SM}(s) = V^{-1}(s)u(s)$ and $y^{SM}(s) = U(s)y(s)$:
\vspace{-0.1cm} \begin{equation*}
\begin{split}
y^{SM}(s)& =  U(s)P(s)u(s) = U(s)P(s)V(s)u^{SM}(s) \\ &\Rightarrow y^{SM}(s) = P^{SM}(s)u^{SM}(s).
\end{split}  \vspace{-0.1cm} 
\end{equation*} Thus, the plant can be represented as in Figure  \ref{fig:blockDiag_csPlantSM2} with $u^{SM}$ and $y^{SM}$ as input and output signals of  $P^{SM}$.
By defining the controller in the format
\vspace{-0.1cm} \begin{equation}\label{eq:CSMtransfC}
C(s) = V(s)C^{SM}(s)U(s),  \vspace{-0.1cm} 
\end{equation} as depicted in Figure  \ref{fig:blockDiag_csPlantContrSM}, the controller $C^{SM}(s)$ can be designed specifically for the diagonal plant $P^{SM}(s)$. By imposing $C^{SM}(s)$ to be also diagonal, the decoupling of the inputs and outputs is accomplished (see \cite{diaz2019pid} for more details).

\begin{figure}
      \centering
		\includegraphics[width = 8.5cm]{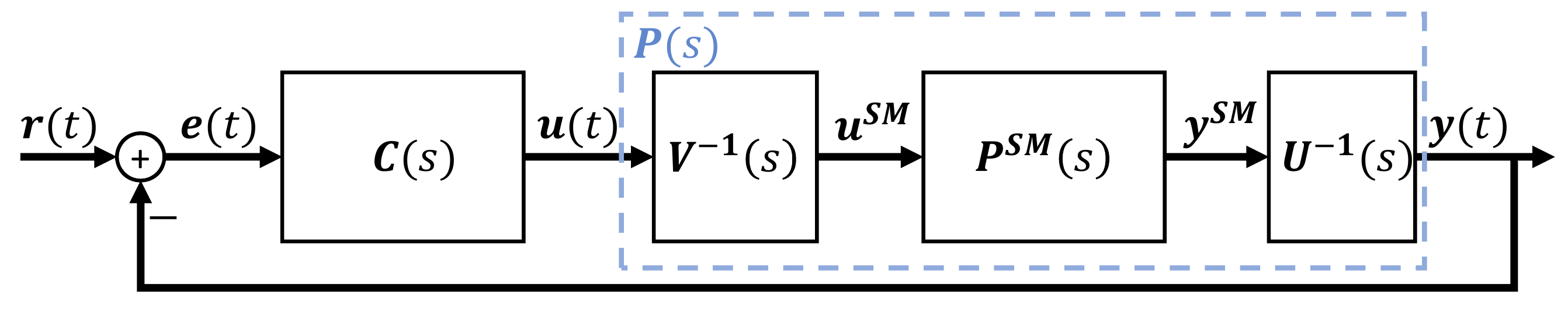} \vspace{-0.2cm}
      \caption{Control system with the plant represented by the Smith-McMillan transformation: $P(s) = U^{-1}(s)P^{SM}(s)V^{-1}(s)$.}
      \label{fig:blockDiag_csPlantSM2}\vspace{-0.2cm}
\end{figure}

The diagonal plant and controller are noted as
\vspace{-0.1cm}\begin{align*}
P^{SM}(s) &= \mathrm{diag} (P^{SM}_1(s),...,P^{SM}_n(s))\\
C^{SM}(s) &= \mathrm{diag} (C^{SM}_1(s),...,C^{SM}_n(s)),
\end{align*} \vskip -0.1cm \noindent where $P^{SM}_i(s)$ and $C^{SM}_i(s)$, $i = 1,...,n$, are scalar real rational transfer functions. A controller can then be designed for each single loop $i$ for the decoupled system, independently, by any SISO control design method.
\begin{figure}
      \centering
		\includegraphics[width = 8.5cm]{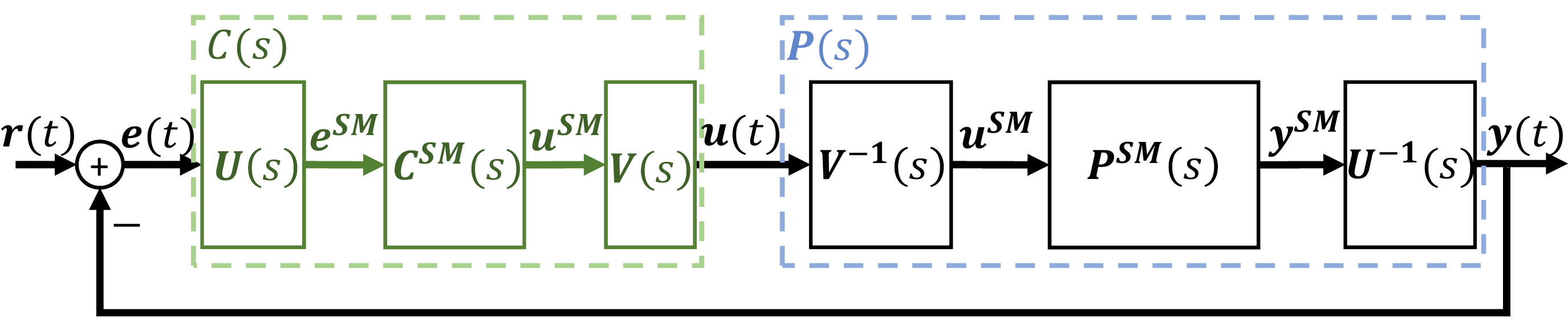} \vspace{-0.2cm}
      \caption{Control system with the plant represented by the Smith-McMillan transformation ($P(s) = U^{-1}(s)P^{SM}(s)V^{-1}(s)$), together with the proposed control structure $C(s) = V(s)C^{SM}(s)U(s)$.}
      \label{fig:blockDiag_csPlantContrSM} \vspace{-0.3cm}
\end{figure}
The control system with $P(s)$ and $C(s)$ is called here \textit{original system} \footnote{The nomenclature essential is adopted because the essential plant consists of the fundamental (essential) dynamic elements of the original plant (the poles and zeros of the transfer matrix), but in a diagonal form.} (Figure  \ref{fig:blockDiag_controlSys}) and the system with $P^{SM}(s)$ and $C^{SM}(s)$ is called \textit{essential system} (Figure  \ref{fig:blockDiag_controlSysSM}) with corresponding transfer matrices, respectively. 
Next, we investigate how this decoupling technique relates to the overall control system design. First, the transformations between domains are presented and, based on them,  the translation of stability and performance requirements between the domains is explored.

\section{TRANSFORMATIONS BETWEEN DOMAINS} 

In this section, the formulas to relate the transfer matrices between the original and essential domains are investigated. This constitutes Contribution (C1).

Assume $P(s)$, $C(s)$, $P^{SM}(s)$ and $C^{SM}(s)$ are fixed real rational proper transfer matrices\footnote{The condition for properness of $C(s)$ based on $C^{SM}(s)$ is given in \cite{mohsenizadeh2015multivariable}.}  and that the closed-loop feedback system interconnections in each domain is well-posed, i.e., all closed-loop transfer matrices are well-defined and proper. From the decoupling procedure,  $U(s)$ and $V(s)$ are unimodular matrices and, as a consequence, $U^{-1}(s)$ and $V^{-1}(s)$ are also polynomial matrices that exist for all $s\in \mathbb{C}$. 
In the sequel the $s$ is omitted  for simplicity and the following matrices properties are used: $(AB)^{-1}= B^{-1} A^{-1}$ and $I = A^{-1}A$, where $A, B \in \mathbb{C}^{n\times n}$ are full rank matrices.

The open-loop transfer matrix is defined as $L\triangleq PC$ for the original and as $L^{SM}\triangleq P^{SM}C^{SM}$ for the essential domain. From (\ref{eq:PtransfPSM}) and (\ref{eq:CSMtransfC}), we have 
\vspace{-0.1cm}\begin{equation*}
\begin{split}
L = (U^{-1}P^{SM}V^{-1})(VC^{SM}U) 
= U^{-1}P^{SM}C^{SM}U  = U^{-1}L^{SM}U.
\end{split}\vspace{-0.2cm}
\end{equation*} The \textit{sensitivity} for the original and essential systems are defined as 
$S \triangleq (I + L)^{-1}$ and $S^{SM} \triangleq (I + L^{SM})^{-1}$. Then, 
\vspace{-0.1cm}\begin{equation*}
\begin{split}
S &= (I + U^{-1}L^{SM}U)^{-1} 
= \left(U^{-1}U + U^{-1}L^{SM}U\right)^{-1} \\
& = (U^{-1}(I + L^{SM})U)^{-1} 
= U^{-1}(I + L^{SM})^{-1}U 
= U^{-1}S^{SM}U.
\end{split}\vspace{-0.1cm}
\end{equation*} Similarly, the other sensitivities transfer matrices are defined and the transformations between domains are presented.
\begin{itemize}
\item \textit{complementary sensitivity} for the original and essential systems: $T = L(I + L)^{-1}$ and $T^{SM} = L^{SM}(I + L^{SM})^{-1}$, respectively. Transformation:
\end{itemize}\begin{equation*}
\begin{split}
T &= \left(U^{-1}L^{SM}U)\right(I + U^{-1}L^{SM}U)^{-1}\\
&= \left(U^{-1}L^{SM}U\right) U^{-1}(I + L^{SM})^{-1}U\\
&= U^{-1}L^{SM}(I + L^{SM})^{-1}U 
= U^{-1}T^{SM}U;
\end{split}\vspace{-0.1cm}
\end{equation*}\begin{itemize}
\item \textit{process sensitivity} for the original and essential systems:  
$S_P \triangleq (I + L)^{-1}P$ and $S_P^{SM} \triangleq (I + L^{SM})^{-1}P^{SM}$, respectively. Transformation: 
\end{itemize}
\begin{equation*}
\begin{split}
S_P &= \left(I + U^{-1}L^{SM}U\right)^{-1} U^{-1}P^{SM}V^{-1}\\
&= U^{-1}(I + L^{SM})^{-1}UU^{-1}P^{SM}V^{-1} \\
&= U^{-1}(I + L^{SM})^{-1}P^{SM}V^{-1} 
= U^{-1}S_P^{SM}V^{-1};
\end{split}\vspace{-0.1cm}
\end{equation*} 
\begin{itemize}
\item \textit{controller sensitivity} for the original and essential systems:
$S_C \triangleq C(I + L)^{-1}$ and $S_C^{SM} \triangleq C^{SM}(I + L^{SM})^{-1}$, respectively. Transformation: 
\end{itemize}\begin{equation*}
\begin{split}
S_C &= VC^{SM}U\left(I + U^{-1}L^{SM}U\right)^{-1} \\
&= VC^{SM}UU^{-1}(I + L^{SM})^{-1}U \\
&= VC^{SM}(I + L^{SM})^{-1}U 
= VS_C^{SM}U.\vspace{-0.1cm}
\end{split}
\end{equation*} 

\begin{figure}
	    \centering
		\includegraphics[width = 7.5cm]{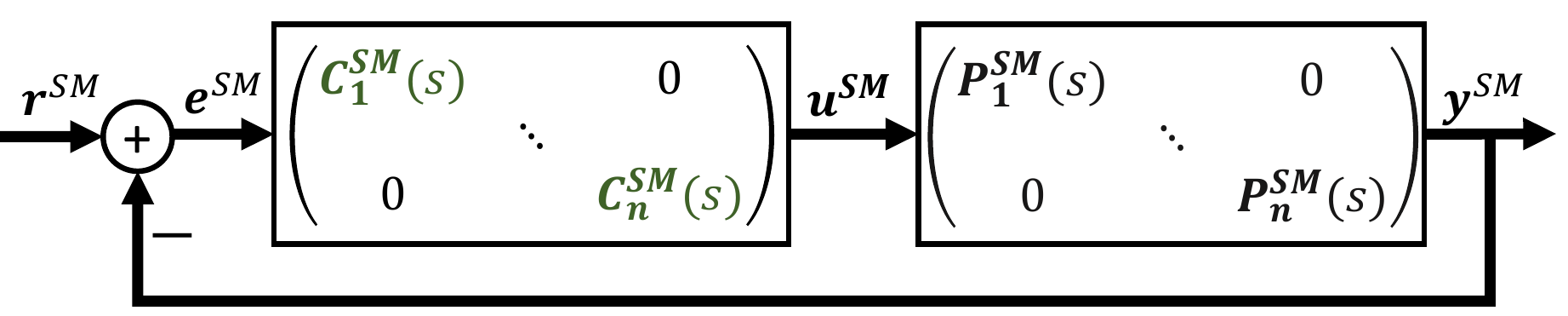} \vspace{-0.2cm}
	    \caption{Decoupled feedback control system (essential system).}
	\label{fig:blockDiag_controlSysSM} \vspace{-0.4cm}
\end{figure}

The presented definitions are regarding transfer matrices evaluated at the output of the plant. The respective input transfer matrices for the essential system are defined as: $L_I^{SM} \triangleq C^{SM}P^{SM}$, $S_I^{SM} \triangleq (I + L_I^{SM})^{-1}$, $T_I^{SM} \triangleq (I + L_I^{SM})^{-1}L_I^{SM}$, $S_{PI}^{SM} \triangleq P^{SM}(I + L_I^{SM})^{-1}$ and $S_{CI}^{SM} \triangleq (I + L_I^{SM})^{-1}C^{SM}$. A similar procedure can be applied to the input open-loop transfer matrices and to the input sensitivities in order to obtain the transformations between domains, as follows.
\begin{itemize}
\item Input loop transfer martrix: $L_I \triangleq CP = (VC^{SM}U)(U^{-1}P^{SM}V^{-1}) 
= VL_I^{SM}V^{-1};$
\item Input sensitivity: $S_I \triangleq \left( I + L_I\right)^{-1} = V S_I^{SM} V^{-1};$
\item Input comp. sens.: $T_I \triangleq \left( I + L_I\right)^{-1}L_I = V T_I^{SM} V^{-1};$ 
\item Input process sens.: $S_{PI} \triangleq P\left( I + L_I\right)^{-1} = U^{-1}PS_{PI}^{SM} V^{-1};$ 
\item Input control sens.: $S_{CI} \triangleq \left( I + L_I\right)^{-1}C = V S_{CI}^{SM} U.$ 
\end{itemize}

Given that $U$, $V$,  $U^{-1}$ and $V^{-1}$ are used to transform from one domain into the other, they are named \textit{transformation matrices}. The relations between domains are used in the next section to determine sufficient conditions for stability of the original system, based on the stability of the essential system.


\section{STABILITY ANALYSIS} 

In this section, the conditions for the stability of the original system are presented in case of a stable closed-loop in the essential domain, which forms Contribution C2. The internal stability concept is analyzed, which guarantees that all signals in the system are bounded provided that the signals injected at any location are bounded  \cite{zhou1998essentials,skogestad2007multivariable}. 

Consider a system described by the standard block diagram in Figure  \ref{fig:blockDiag_csPlantSM3} and assume that the system is well-posed\footnote{Well-posedness is equivalent to that $S$ exists and is proper \cite{zhou1998essentials}.}. Notice that it is equivalent to the system in Figure  \ref{fig:blockDiag_csPlantContrSM}, but with all possible inputs and outputs considered.

\begin{figure}
      \centering
		\includegraphics[width = 6cm]{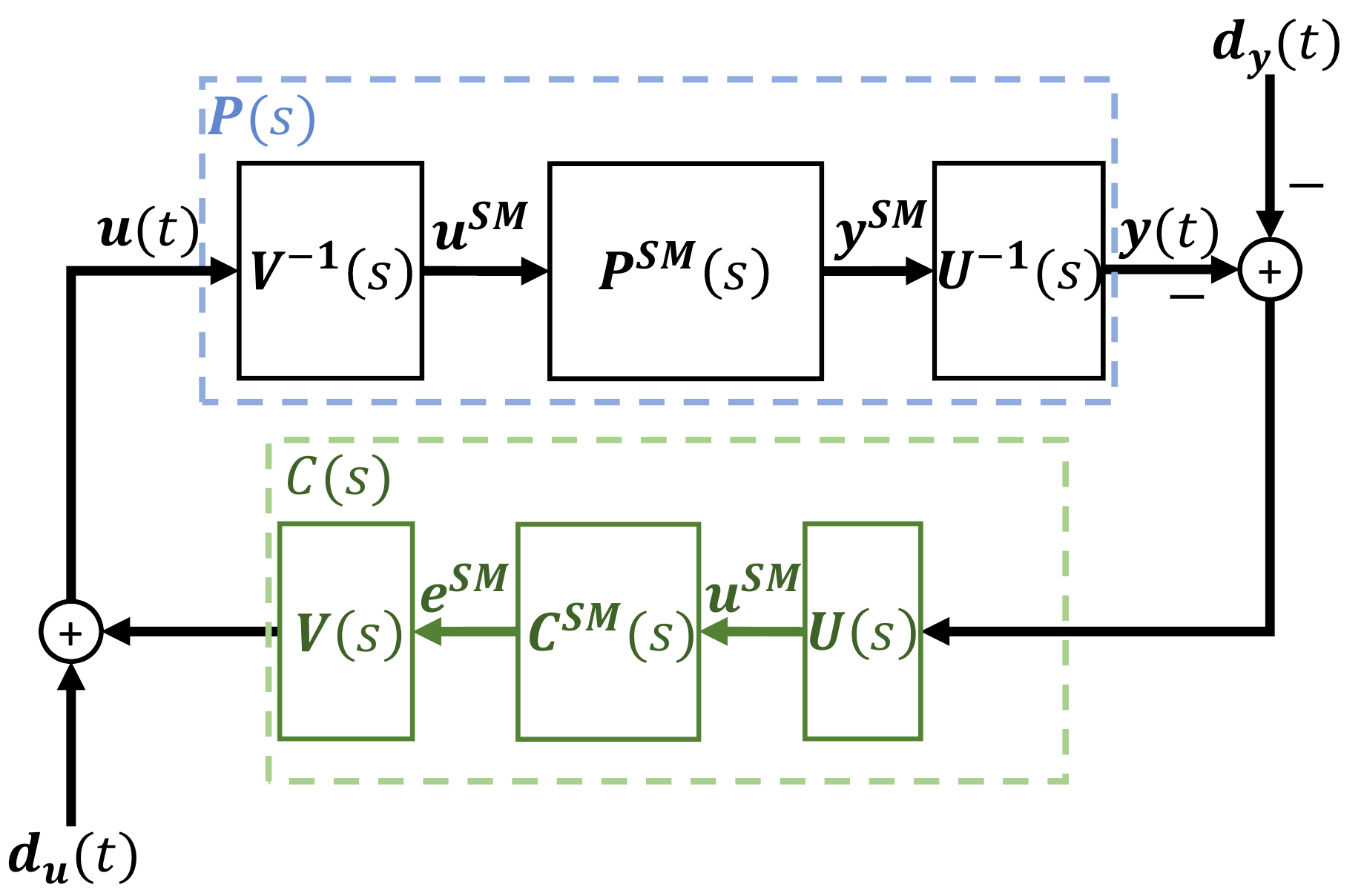} \vspace{-0.2cm}
      \caption{Block diagram used to check internal stability of a feedback system.}
      \label{fig:blockDiag_csPlantSM3}\vspace{-0.4cm}
\end{figure}

\begin{definition}\textit{(Internal stability)} The system of Figure  \ref{fig:blockDiag_csPlantSM3} is said to be internally stable if it is well-posed and \vspace{-0.1cm}
\begin{equation}\label{eq:internStabTestMatrix}
\begin{bmatrix}
(I +CP)^{-1} & - C(I +PC)^{-1}\\
P(I +CP)^{-1} &- (I +PC)^{-1}
\end{bmatrix}
\end{equation}
belongs to $\mathcal{R}\mathcal{H}_\infty$.
\end{definition} 

\begin{theorem}\label{th:internalStability}
The original system in Figure  \ref{fig:blockDiag_csPlantContrSM} is  internally stable if it is well-posed and if the essential feedback system in Figure \ref{fig:blockDiag_controlSysSM} is internally stable.
\end{theorem} 

\begin{proof}
If the essential system is internally stable, it is well-posed and the closed-loop transfer matrix 
\begin{equation}\label{eq:internStabProofMatrixSM}
\begin{bmatrix}
(I +C^{SM}P^{SM})^{-1} & - C^{SM}(I +P^{SM}C^{SM})^{-1}\\
P^{SM}(I +C^{SM}P^{SM})^{-1} &- (I +P^{SM}C^{SM})^{-1}
\end{bmatrix}
\end{equation} belongs to $\mathcal{R}\mathcal{H}_\infty$. 
Consider that the original system is well-posed. Using the sensitivities transformations from the essential to the original domain, we have that
\vspace{-0.2cm}\begin{multline}\label{eq:internStabProofMatrix}
		\left.
		\begin{bmatrix}
(I +CP)^{-1} & - C(I +PC)^{-1}\\
P(I +CP)^{-1} &- (I +PC)^{-1}
\end{bmatrix} =
		\right. \\
		\left. 
		\Bigg[ \begin{matrix}
V(I +C^{SM}P^{SM})^{-1}V^{-1} \\
U^{-1}P^{SM}(I +C^{SM}P^{SM})^{-1}V^{-1}
\end{matrix} 
\begin{matrix}
 - V C^{SM}S^{SM}U\\
-U^{-1}S^{SM}U
\end{matrix} \Bigg], \right. \vspace{-0.1cm}
		\end{multline}
with $S^{SM} = (I +P^{SM}C^{SM})^{-1}$. Because $U$, $V$, $U^{-1}$ and $V^{-1}$ are polynomial matrices on $s$, the transformation into the original domain only adds zeros to elements of  (\ref{eq:internStabProofMatrixSM}) and it does not change the pole locations or add new poles. Poles can be canceled, but as all poles are in the left half part of the complex plane, internal stability is not affected. 
Thus, the test matrix in (\ref{eq:internStabProofMatrix}) also belongs to $\mathcal{R}\mathcal{H}_\infty$ and the original system is internally stable by definition.
\end{proof}

Observe that if the essential system is unstable, it does not mean that the original system is also unstable, because the transformation from essential to original domain can lead to pole-zero cancellations and the unstable poles from the essential system can be canceled.

\begin{corollary}\label{th:internalStabilityCorollary}
Assume that $C^{SM}$ is a diagonal transfer matrix. The original system in Figure  \ref{fig:blockDiag_csPlantContrSM} is internally stable if it is well-posed and if the decoupled SISO control systems with $C^{SM}_i(s)$ and $P^{SM}_i(s)$ are internally stable for all $i \in [1,n]$.
\end{corollary} 

\begin{proof}
Given that $C^{SM}$ is diagonal, the essential system is completely decoupled and the internal stability of the overall MIMO system is guaranteed if all the SISO systems it comprises are internally stable. Thus, based on Theorem~ \ref{th:internalStability}, internal stability of the original system is guaranteed.
\end{proof}

The guarantee of internal stability for the original system, given the conditions on Theorem \ref{th:internalStability}, encourages further investigation of this decoupling strategy. Next, the translation of performance requirements is addressed. 
 
\section{PERFORMANCE ANALYSIS}

The closed-loop performance requirements are defined in terms of magnitude limits for the maximum singular value of the closed-loop transfer matrices at a specified frequency range. We analyze how the original system's specifications can be translated to the essential system (Contribution C3).

In the sequel, the analysis is performed regarding the sensitivity transfer matrix $S$, but can analogously be extended to other closed-loop transfer matrices. Consider a proper linear stable transfer matrix $M(s)$. The $\mathcal{H}_\infty$ norm of $M$ is defined as $\parallel M(j\omega)\parallel_\infty \triangleq \max\limits_\omega \bar{\sigma}(M(j\omega))$, where $\bar{\sigma}(M(j\omega))$ is the maximum singular value of the complex matrix $M(j\omega)$ for a given $s = j\omega$, $\omega\in \mathbb{R}$. 

First, consider that the performance requirement for $S$ is given in the format
\vspace{-0.2cm}\begin{equation}\label{eq:performSrequirement}
\parallel w_1 S\parallel_\infty \leq 1, \vspace{-0.1cm}
\end{equation}where $w_1$ is a scalar filter for which the absolute value of  $1/w_1(j\omega)$ represents the maximum magnitude allowed for $\bar{\sigma}(S(j\omega))$ at each frequency $\omega$; in general, the scalar weighting $w_1(s)$ can be replaced by a matrix $W_1(s)\in \mathbb{C}^{n\times n}$. This type of requirement is commonly used, for instance, in mixed-sensitivity shaping controller synthesis \cite{skogestad2007multivariable}.

To guarantee that the performance requirements are fulfilled, an approach is to use the transformations of the sensitivities to translate the requirements to the essential domain. By substituting $S = U^{-1}S^{SM}U$ in (\ref{eq:performSrequirement}) and applying the submultiplicative property, we obtain that
\vspace{-0.1cm}\begin{equation*}
\parallel w_1 S\parallel_\infty \leq \parallel w_1\parallel_\infty\parallel U^{-1}\parallel_\infty\parallel S^{SM}\parallel_\infty\parallel U\parallel_\infty;
\end{equation*}thus,  (\ref{eq:performSrequirement}) is satisfied if we assure that 
\vspace{-0.1cm}\begin{equation*}
\parallel w_1\parallel_\infty\parallel U^{-1}\parallel_\infty\parallel S^{SM}\parallel_\infty\parallel U\parallel_\infty< 1.
\end{equation*} 
Notably, as $U$, $U^{-1}$, $V$ and $V^{-1}$ are polynomial matrices of $s$, their $\mathcal{H}_\infty$ norm only exists when they are constant matrices (properness condition), which is usually not the case. Two ideas to circumvent this problem are: to limit the performance specifications to a set of frequencies or; to modify the $U$ and $V$ matrices in a way that they become proper. The second option is not analyzed here because the decoupling is degraded. 

By limiting the performance specifications to a set of frequencies $\omega_0\leq\omega\leq\omega_f$ ($\omega_0,\omega,\omega_f\in \mathbb{R}$), we have that the maximum singular values of the transformation matrices are defined for each $\omega$ in this interval. Thus, the performance  specification can be defined as
\vspace{-0.1cm}\begin{equation*}
\bar{\sigma}(w_1(j\omega)S(j\omega)) < 1 \; \forall \omega\in (\omega_0,\omega_f).
\end{equation*} Applying the matrix transformation between domains and the submultiplicative property, we have
\vspace{-0.1cm}\begin{equation*}
\begin{split}
\bar{\sigma}(w_1(j\omega)S(j\omega)) &= \bar{\sigma}(w_1(j\omega)U^{-1}(j\omega)S^{SM}(j\omega)U(j\omega))
\end{split} \vspace{-0.1cm}
\end{equation*}
\begin{flushright}
$\leq \bar{\sigma}(w_1(j\omega))\bar{\sigma}(U^{-1}(j\omega))\bar{\sigma}(S^{SM}(j\omega))\bar{\sigma}(U(j\omega))
$.
\end{flushright} The performance requirement is satisfied if
\vspace{-0.1cm}\begin{equation}\label{eq:performSrequirementLim}
\bar{\sigma}(S^{SM}(j\omega)) \leq \frac{1}{\bar{\sigma}(w_1(j\omega))\bar{\sigma}(U^{-1}(j\omega))\bar{\sigma}(U(j\omega))}.\vspace{-0.1cm}
\end{equation}

The obtained inequality (\ref{eq:performSrequirementLim}) is a sufficient but conservative condition. For high frequencies, for instance, if $U$ is not a constant matrix, then $\bar{\sigma}(U(j\omega))$ and $\bar{\sigma}(U^{-1}(j\omega))$ have high magnitude. In addition, if usual control design requirements are considered, the maximum magnitude of $1/w_1(j\omega)$ is $1/2$ ($-6$ dB)  \cite{skogestad2007multivariable}. Conversely, for proper transfer function elements in $L^{SM}$, $\bar{\sigma}(S^{SM}(j\omega))$ is approximately one. These expected magnitudes contradict the inequality in (\ref{eq:performSrequirementLim}), indicating that the conservativeness introduced by the proposed inequality is impracticable for high frequencies. 

The conservative characteristic is introduced mostly by the submultiplicative property: the effect of the poles of $S^{SM}$ along the frequency from the designed system cancels the effect of the zeros of the transformation matrices while converting to $S$, but these cancellation effects are not present when separating the maximum singular values of each matrix for the inequality. Although the condition is excessively conservative for some frequencies, it can still be useful for a specific range. Possibly at low/medium frequency ranges for $S$. For the complementary sensitivity $T$, there can be compatibility issues caused by the conservative characteristic of the inequality at low frequencies, as demonstrated next in the example. 

In conclusion, the general non-properness of the transformation matrices restrains the use of $\mathcal{H}_\infty$ norm together with the submultiplicative property. The maximum singular value norm at a specific frequency range is proposed to circumvent this limitation, but it may lead to excessively restrictive conditions when combined to the submultiplicative property. To have less conservative conditions, first the transformations can be performed and afterwards the performance specifications applied, such as 
$\parallel w_1 S\parallel_\infty = {\parallel w_1 U^{-1}S^{SM}U \parallel_\infty \leq 1}$, but it possibly results in complex requirements for the essential system. Another option is to perform the domain transformation and have specific conditions for each term of the obtained matrices, as demonstrated next in the example. 



\section{EXAMPLE}

The Smith McMillan decoupling procedure, combined with manual control tuning for the decoupled essential system, is applied to a simple example of a mechanical system. The following one-dimensional mechanical system is considered: two masses, two springs, two dampers and two input forces, arranged as depicted in Figure \ref{fig:rigidBodySystem}. 
\begin{figure}
      \centering
		\includegraphics[width = 7cm]{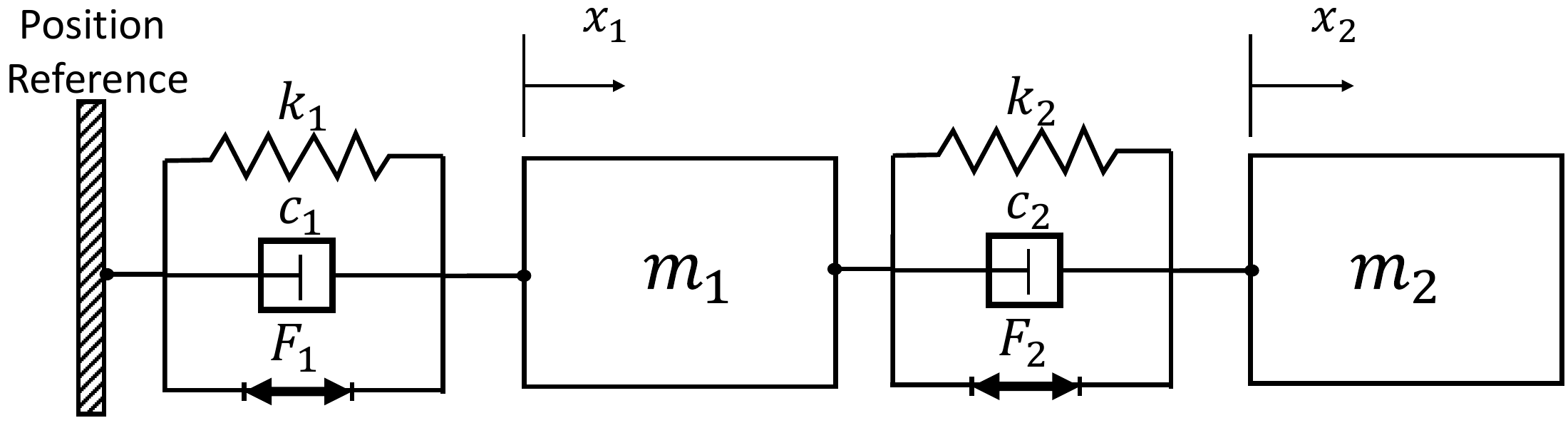} 
      \caption{Schematic view of the plant for the example. Two mass system with two input forces, $F_1$ and $F_2$, and two positions measured, $x_1$ and $x_2$.}
      \label{fig:rigidBodySystem}\vspace{-0.4cm}
\end{figure}

The system has two inputs, $F_1(t)$ and $F_2(t)$, and the outputs chosen are $x_1(t)$ and $x_2(t)$. The differential equations that describe the system dynamics are:
\begin{equation*}
\begin{split}
m_1\ddot{x}_1(t) = & \: F_1(t) - F_2(t) - k_1x_1(t)- c_1\dot{x}_1(t)\\
& - k_2(x_1(t) - x_2(t))- c_2(\dot{x}_1(t) - \dot{x}_2(t)),\\
m_2\ddot{x}_2(t) = & \: F_2(t) + k_2(x_1(t) - x_2(t)) + c_2(\dot{x}_1(t) - \dot{x}_2(t)),
\end{split}
\end{equation*} being respectively $\dot{x}$ and $\ddot{x}$ the first and second derivative of $x$ with respect to time $t$, $k_1$ and $k_2$ are the springs stiffness, and $c_1$ and $c_2$ are the dampers damping coefficient. Table~\ref{tab:table_example} contains the considered values for the coefficients of the model which are used onward. 
\vspace{-0.2cm}
\begin{table}[h]
\caption{Values of the coefficients used in the model}
\vspace{-0.1cm}
\label{tab:table_example}
\centering
\begin{tabular}{|c|c|c|c|c|c|}
\hline
$m_1$ & $m_2$ & $k_1$ & $c_1$ & $k_2$ & $c_2$\\
\hline
$1$ kg & $1$ kg & $10$ N/m & $10$ kg/s& $10$ N/m & $10$ kg/s\\
\hline
\end{tabular}
\end{table}
\vspace{-0.3cm}

By computing the Laplace transform and isolating the positions $x_1$ and $x_2$, we have that $ X (s)= P(s)\:F(s) $ with
\begin{equation*}
P(s) = 
\begin{bmatrix}
\frac{s^2 + 10s + 10}{s^4 + 30s^3  + 130s^2  + 200 s + 100} & -\frac{s^2}{s^4 + 30s^3  + 130s^2  + 200 s + 100} \\
\frac{10 + 10s}{s^4 + 30s^3  + 130s^2  + 200 s + 100} & \frac{s^2+10s+10}{s^4 + 30s^3  + 130s^2  + 200 s + 100}
\end{bmatrix},
\end{equation*} 
\begin{equation*}
X(s) = \begin{bmatrix}
X_1(s) \\ X_2(s)
\end{bmatrix} \quad \mathrm{and} \quad
F(s) = \begin{bmatrix}
F_1(s) \\ F_2(s)
\end{bmatrix}.
\end{equation*}
The $X_1(s)$, $X_2(s)$, $F_1(s)$ and $F_2(s)$ represents the Laplace transforms of $x_1(t)$, $x_2(t)$, $F_1(t)$ and $F_2(t)$, respectively.

The Smith-McMillan form of $P$ is 
\begin{equation*}
P^{SM}(s) = 
\begin{bmatrix}
\frac{1}{s^4 + 30s^3  + 130s^2  + 200 s + 100} & 0 \\
0 & 1
\end{bmatrix},
\end{equation*} computed by using symbolic variables with the \textit{smithForm()} function on Matlab. The $(s)$ is omitted in the sequel for simplicity. The transformation $P^{SM} = U\:P\:V$ is performed by the following unimodular matrices:
\begin{equation*}
U = 
\begin{bmatrix}
0 & 1\\
1 & \frac{s^3 + 29s^2 +100s + 90}{10}
\end{bmatrix}
\mathrm{,} \quad
V = 
\begin{bmatrix}
- \frac{(s+9)}{10} & s^2 + 10s + 10\\
1 & -10s-10
\end{bmatrix}.
\end{equation*}

The main control objective is to have each mass following a specified reference trajectory. Before designing the controllers, it is necessary to determine the lowest relative order permitted for properness of the original system's controller.

The controller for the essential system is defined as $ C^{SM}(s) = \mathrm{diag} (C^{SM}_1(s),C^{SM}_2(s))$, with $C_{1}^{SM}$ and $C_{2}^{SM}$ to be designed. The original system's controller is then obtained by $C = VC^{SM}U$ and given by $ C = \begin{bmatrix} C_{c1} & C_{c2}
\end{bmatrix}$, with
$ {C_{c1} =
\begin{bmatrix}
(s^2 + 10s + 10)C_{2}^{SM} \\
(-10s-10)C_{2}^{SM}
\end{bmatrix}}$ and
\begin{equation*}
C_{c2} =
\begin{bmatrix}
- \frac{(s+9)}{10}C_{1}^{SM} + \frac{(s^2 + 10s + 10)(s^3 + 29s^2 +100s + 90)}{10}C_{2}^{SM}\\
C_{1}^{SM} +\frac{(-10s-10)(s^3 + 29s^2 +100s + 90)}{10}C_{2}^{SM}
\end{bmatrix}.
\end{equation*} 
The $C$ matrix was divided in columns matrices just to fit properly in this paper format. In order to have the MIMO controller as a proper transfer matrix, $C_1^{SM}$ and $C_2^{SM}$ must have the order of the denominator polynomial in $s$ higher by $1$ and $5$, respectively, than the order of the numerator. The necessity of a relative degree for the design of each controller, specially $C_2^{SM}$, can be challenging. 

Both controllers for the Smith-McMillan plant are designed to have a cross-over frequency of $10$ Hz and following the guidelines for the minimum relative degree,  as follows: 
\begin{itemize}
\item For $C_1^{SM}$, the plant is $P_1^{SM} = \frac{1}{s^4 + 30s^3  + 130s^2  + 200 s + 100}$ and the needed relative degree $1$. The controller is composed of an integrator for zero steady state error and three lead filters for stability;
\item For $C_2^{SM}$, the plant is $P_2^{SM} = 1$ and the needed relative degree $5$. The controller is composed of a second order integrator and three poles for the relative order, and a lead filter for stability. 
\end{itemize} Figure \ref{fig:blockDiag_exampleLsm} contains the obtained open-loop transfer functions for each SISO system. This design is noted as \textbf{Design 1}.


Reference tracking is guaranteed as proven in \cite{mohsenizadeh2015multivariable}. Stability is also guaranteed because each SISO system is internally stable. However, without taking into consideration the effect of the transformation matrices, poor performance is achieved by the control system in the original domain, as depicted in Figure \ref{fig:blockDiag_exampleT}: the magnitude plot from the reference $r_2$ to the position $x_1$ indicates that a change in the reference signal for the position of the mass $2$ has a significant impact in the transient response of the mass $1$ position.

\begin{figure}
      \centering
		\includegraphics[width = 6cm]{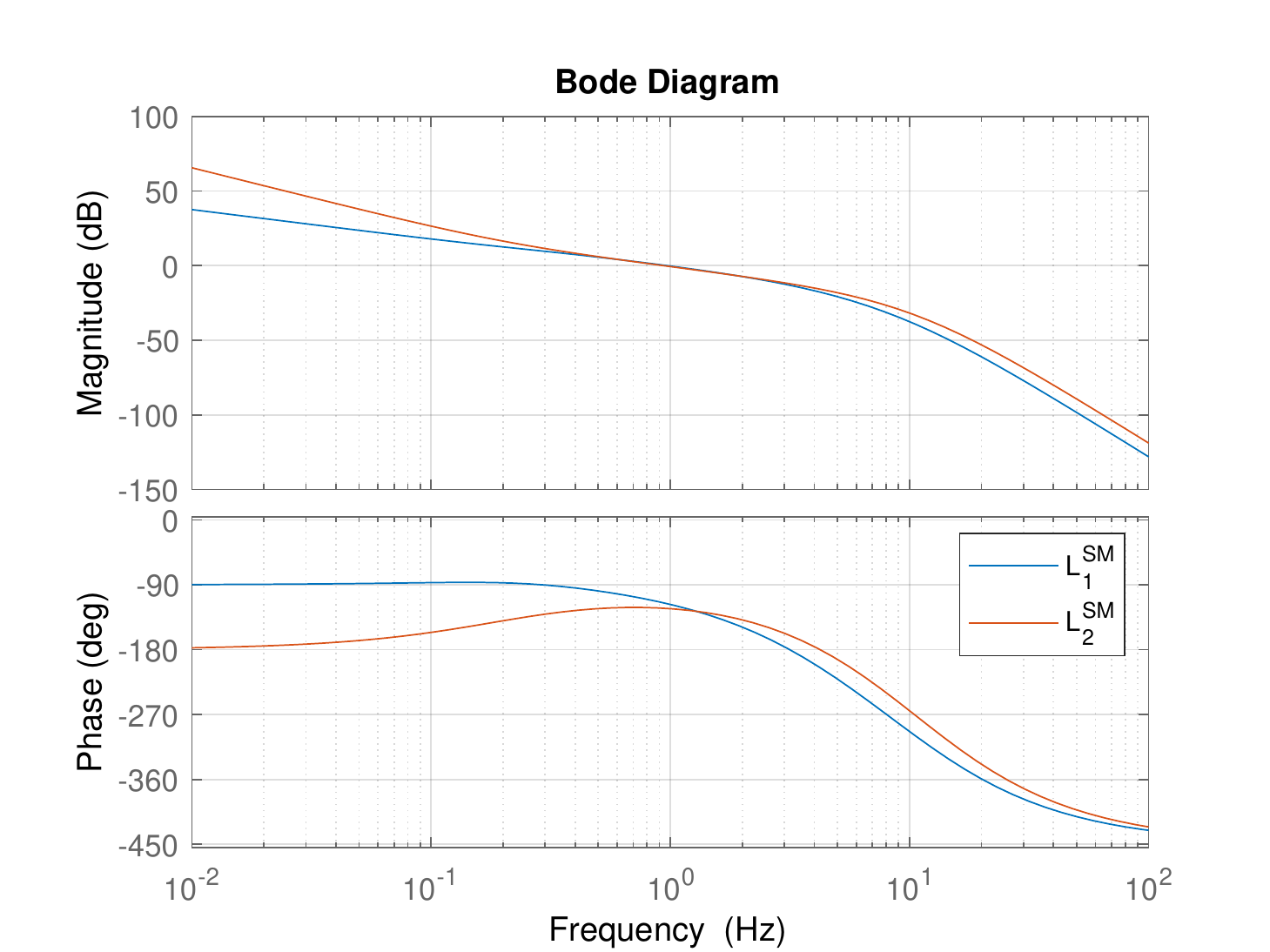} \vspace{-0.2cm}
      \caption{Bode plot of $L_1^{SM}$ (blue curve) and $L_2^{SM}$ (orange curve).}
      \label{fig:blockDiag_exampleLsm} \vspace{-0.4cm}
\end{figure}

\begin{figure}
      \centering
		\includegraphics[width = 6.5cm]{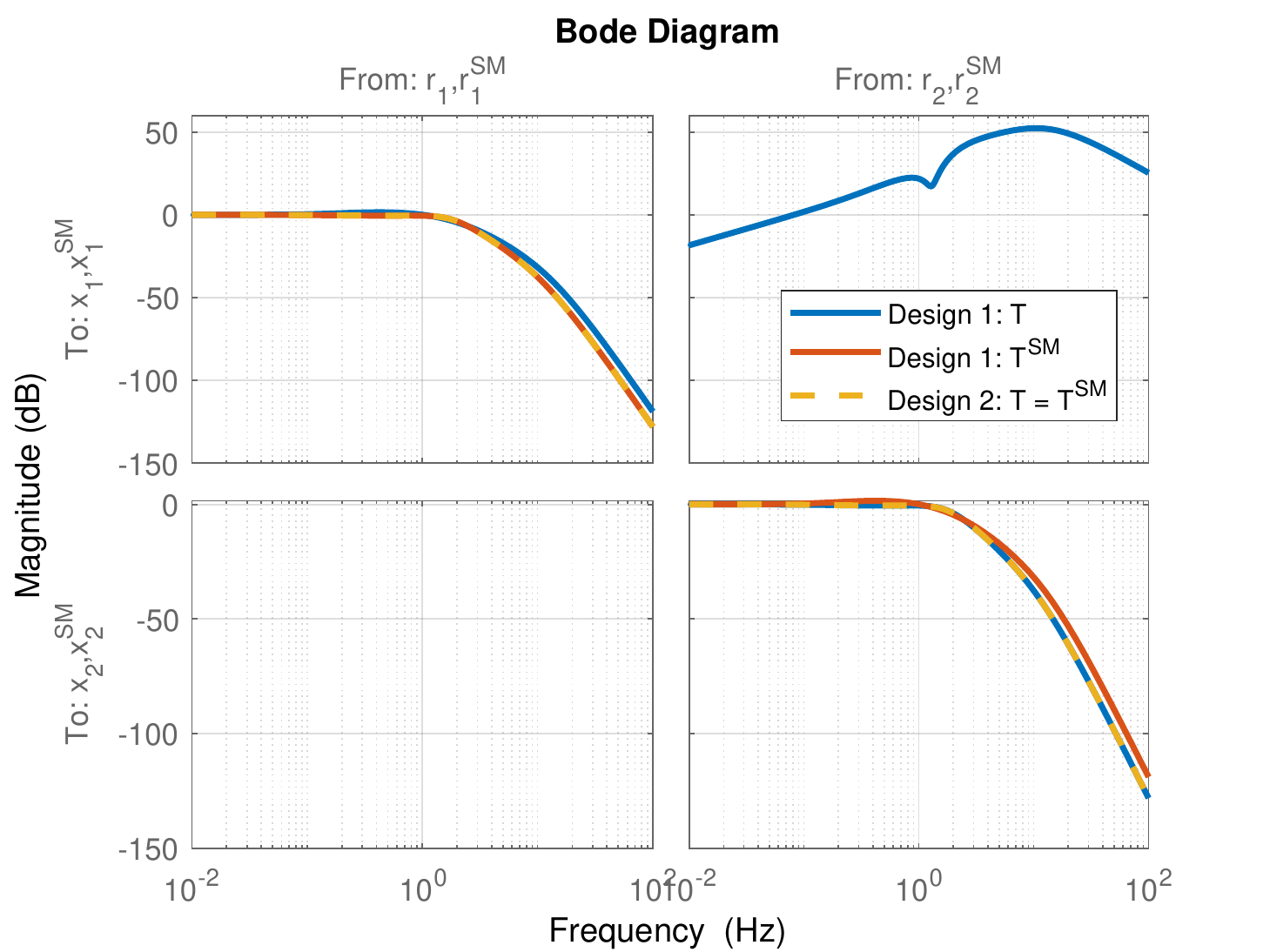} \vspace{-0.2cm}
      \caption{Magnitude Bode plot of $T$ (blue curve), $T^{SM}$ (orange curve) for the first design and $T = T^{SM}$ (yellow curve) for the second design.}
      \label{fig:blockDiag_exampleT}\vspace{-0.4cm}
\end{figure} 

Defining the performance requirement for $\omega\in (10^{-2},10^2)$ by $\bar{\sigma}(w_2(j\omega)T(j\omega)) < 1$ with $w_2(j\omega) = 1/1.25$ (value based on \cite{skogestad2007multivariable}), it is satisfied if 
$ \bar{\sigma}(T^{SM}(j\omega)) \leq \frac{1.25}{\bar{\sigma}(U^{-1}(j\omega))\bar{\sigma}(U(j\omega))}$. The issues caused by the conservative characteristic of the inequality for low frequencies is clear from the plot in Figure \ref{fig:blockDiag_performanceCondition}: to fulfill the requirement, $T_1^{SM}(j\omega)$ and $T_2^{SM}(j\omega)$ would need to have a magnitude less then $0.015$ for low frequencies, which deteriorates the reference tracking capability. 

\vspace{-0.3cm}
\begin{figure}[thpb]
      \centering
		\includegraphics[width = 6cm]{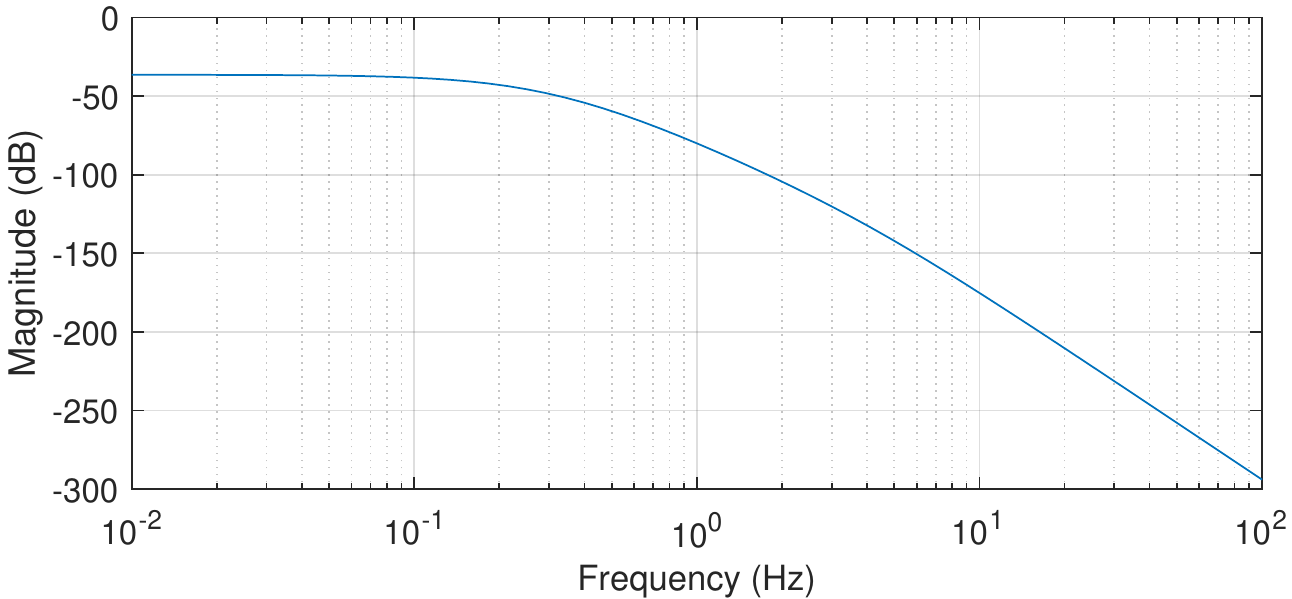} \vspace{-0.2cm}
      \caption{Magnitude plot of  $\frac{1.25}{\bar{\sigma}(U^{-1}(j\omega))\bar{\sigma}(U(j\omega))}$ over frequencies from $10^{-2}$ Hz to $10^2$ Hz: performance condition for $T^{SM}$.}
      \label{fig:blockDiag_performanceCondition}
\end{figure}
\vspace{-0.3cm} 

To have less restrictive conditions, one approach is to first perform the domain transformation and then specify conditions for each term in the matrix. From ${T(s) = U^{-1}(s)T^{SM}(s)U(s)}$, we obtain 
\vspace{-0.1cm}\begin{equation*}
T(s) = 
\begin{bmatrix}
T^{SM}_2(s) & \frac{s^3 + 29s^2 +100s + 90}{10}(-T^{SM}_1(s)+T^{SM}_2(s))\\
0 & T^{SM}_1(s)
\end{bmatrix}.\vspace{-0.1cm} 
\end{equation*} The pursued specification now is ${\parallel T_{ij}(j\omega)\parallel_\infty \leq 1.25}$ for all $i,j = 1,2$. The $\mathcal{H}_\infty$ norm can be applied here as $T(s)$ has to be composed by proper transfer functions. For this specific example, the obtained expression for $T(s)$ shows that by designing $T^{SM}_1(s) = T^{SM}_2(s)$, there will be no impact on the mass block $1$ position when changing the reference signal for the position of the mass block $2$ and ${\parallel T_{12}(j\omega)\parallel_\infty \leq 1.25}$. In addition, we have $T(s) = T^{SM}(s)$.


By maintaining the same $C_1^{SM}$ as in the previous design, the same $L^{SM}_1(s)$ and $T^{SM}_1(s)$ are obtained and $\parallel T_{22}(j\omega)\parallel_\infty \leq 1.25$ is satisfied. To have $T^{SM}_1(s) = T^{SM}_2(s)$, $L^{SM}_2(s)$ must be equal to $L^{SM}_1(s)$; as $P_2^{SM} = 1$, then $C_2^{SM} = L^{SM}_1(s)$. The new $C_2^{SM}$ satisfies the properness condition and, with this controller, $\parallel T_{11}(j\omega)\parallel_\infty \leq 1.25$ and the performance specification is achieved. This design is noted as \textbf{Design 2}.

The new magnitude Bode plot of $T$ is depicted in Figure \ref{fig:blockDiag_exampleT} (yellow dashed plots). Besides reference tracking and internal stability, performance is achieved in the sense of a change of the reference of one mass does not disturb the reference tracking of the other mass. 

This example showed that by designing the controller for the decoupled/essential system by only considering the decoupled plant, reference tracking and stability are guaranteed for the original control system, but performance of the decoupled system is not. For this specific example, the use of maximum singular value of the performance requirement and the transformation matrices separated by the submultiplicative property is proven excessively conservative. On the other hand, by performing the transformation first and after imposing the performance requirement for the design, 
a direction on how to design the controller to fulfill the performance requirement was obtained. This last strategy is completely dependent on the plant, as mentioned before, and can lead to complex performance requirements for the essential control system design.

\section{CONCLUSION}


The Smith-McMillan form decoupling compensator synthesis is presented. It is proven that the internal stability of the decoupled (essential) system is passed to the original system. Additionally, to guarantee performance, requirements defined on the maximum singular value of the closed-loop transfer matrices for the original system are translated to the essential domain. However, as commented and further demonstrated in the example, the presented performance requirement sufficient condition can be excessively restrictive. 

Also demonstrated in the example, a different strategy for translating the performance requirements can be applied. For this particular case, simple less restrictive conditions were obtained. However, the obtained performance condition strongly depends on the plant's transfer matrix and, for different applications, can result into complex requirements.

Future research steps include the investigation of classes of systems that benefit from this decoupling compensator design technique and its applicability for the definition of types and locations of sensors and actuators in system's design.


\addtolength{\textheight}{-12cm}   

%

\bibliographystyle{IEEEtran}
\bibliography{IEEEabrv,bibSmithMcMillanDec}

\end{document}